\documentclass[11pt,letter]{article}
\usepackage[svgnames]{xcolor}
\usepackage{tikz}
\usepackage[utf8]{inputenc}
\usepackage[T1]{fontenc}
\usepackage{stmaryrd}
\usepackage{textcomp,amssymb,amsmath,amsthm}
\makeatletter
\newtheorem*{rep@theorem}{\rep@title}
\newcommand{\newreptheorem}[2]{%
\newenvironment{rep#1}[1]{%
 \def\rep@title{#2 \ref{##1}}%
 \begin{rep@theorem}}%
 {\end{rep@theorem}}}
\makeatother

\makeatletter
\newtheorem*{rep@proposition}{\rep@title}
\newcommand{\newrepproposition}[2]{%
\newenvironment{rep#1}[1]{%
 \def\rep@title{#2 \ref{##1}}%
 \begin{rep@proposition}}%
 {\end{rep@proposition}}}
\makeatother

\makeatletter
\newtheorem*{rep@corollary}{\rep@title}
\newcommand{\newrepcorollary}[2]{%
\newenvironment{rep#1}[1]{%
 \def\rep@title{#2 \ref{##1}}%
 \begin{rep@corollary}}%
 {\end{rep@corollary}}}
\makeatother

\newtheorem{theorem}{Theorem}
\newreptheorem{theorem}{Theorem}
\newtheorem{proposition}{Proposition}
\newrepproposition{proposition}{Proposition}
\newtheorem{claim}{Claim}
\newtheorem{definition}{Definition}
\newtheorem{lemma}{Lemma}
\newtheorem{corollary}{Corollary}
\newrepcorollary{corollary}{Corollary}
\newcommand{\ktt}{\ensuremath{K_{2,3}}}

\tikzset{smallvertex/.style={
    fill=black,
    circle,
    inner sep=0pt,
    minimum size = 4pt
}}
\tikzset{bigvertex/.style={
    fill=black,
    circle,
    inner sep=0pt,
    minimum size = 6pt
}}
\tikzset{terminal/.style={
    fill=black,
    rectangle,
    inner sep=0pt,
    minimum size = 6pt
}}

\title{When Do Gomory-Hu Subtrees Exist?}
\author{Guyslain Naves
\footnote{guyslain.naves@lis-lab.fr, Aix Marseille University, Université de Toulon, CNRS, LIS, Marseille, France}
, Bruce Shepherd
\footnote{fbrucesh@cs.ubc.ca, University of British Columbia, Vancouver, Canada}
}
\date{\today}

\begin{document}

\maketitle

\begin{abstract}
  Gomory-Hu (GH) Trees are a classical sparsification technique for
  graph connectivity. It is one of the fundamental models in
  combinatorial optimization which also continually finds new
  applications, most recently in social network analysis. For any
  edge-capacitated undirected graph $G=(V,E)$ and any subset of {\em
    terminals} $Z \subseteq V$, a Gomory-Hu Tree is an
  edge-capacitated tree $T=(Z,E(T))$ such that for every $u,v \in Z$,
  the value of the minimum capacity $uv$ cut in $G$ is the same as in
  $T$. Moreover, the minimum cuts in $T$ directly identify (in a
  certain way) those in $G$.  It is well-known that we may not always
  find a GH tree which is a subgraph of $G$.  For instance, every GH
  tree for the vertices of $K_{3,3}$ is a $5$-star. We characterize
  those graph and terminal pairs $(G,Z)$ which always admit such a
  tree. We show that these are the graphs which have no
  \emph{terminal-$K_{2,3}$ minor}. That is, no $K_{2,3}$ minor whose
  vertices correspond to terminals in $Z$. We also show that the
  family of pairs $(G,Z)$ which forbid such $K_{2,3}$ ``$Z$-minors''
  arises, roughly speaking, from so-called Okamura-Seymour instances.
  More precisely, they are subgraphs of {\em $Z$-webs}. A $Z$-web is
  built from planar graphs with one outside face which contains all
  the terminals and each inner face is a triangle which may contain an
  arbitrary graph. This characterization yields an additional
  consequence for multiflow problems. Fix a graph $G$ and a subset
  $Z \subseteq V(G)$ of terminals. Call $(G,Z)$ {\em cut-sufficient}
  if the cut condition is sufficient to characterize the existence of
  a multiflow for any demands between vertices in $Z$, and any edge
  capacities on $G$. Then $(G,Z)$ is cut-sufficient if and only if it
  is terminal-$K_{2,3}$ free.
\end{abstract}

\section{Introduction}

The notion of sparsification is ubiquitous in applied mathematics and
combinatorial optimization is no exception. For instance, shortest
paths to a fixed root vertex in a graph $G=(V,E)$ are usually stored
as a tree directed towards the root. Another classical application is
that of Gomory-Hu (GH) Trees \cite{gomory1961multi} which encode all
of the minimum cuts of an edge-capacitated undirected graph $G=(V,E)$,
with capacities $c: E \rightarrow \mathbb{R}^+$. For each $s,t \in V$, we
denote by $\lambda(s,t)$ the capacity of a minimum cut separating $s$
and $t$. Equivalently $\lambda(s,t)$ is the maximum flow that can be
sent between $s,t$ in $G$ with the given edge capacities. Gomory and
Hu showed that one may encode the $O(n^2)$ minimum cuts by a tree on
$V$.

A \emph{spanning edge-capacitated tree} for $G$ is a spanning tree
$T = (V,E')$ together with a capacity function
$c' : E' \rightarrow \mathbb{R}^+$. Any edge $e \in E'$ induces a
\emph{fundamental cut} $G(A,B)$, where $A$ and $B$ are the vertex set
of the two components of $T \setminus e$. Here we use $G(A,B)$, or
occasionally $\delta(A)=\delta(B)$, to denote the associated cut in
$G$, that is, $G(A,B) = \{e \in E(G) : \mbox{$e$ has one endpoint in
  $A$ and the other in $B$}\}$.

\begin{definition}
  \label{defn:encoding}
  Let $T$ be a spanning edge-capacitated tree.  An edge
  $e=ab \in E(T)$ is \emph{encoding} if its fundamental cut $G(A,B)$
  is a minimum $ab$-cut and its capacity is $c'(e)$, that is,
  $c(G(A,B))=c'(e)$.
\end{definition}

A \emph{Gomory-Hu tree} (GH tree for concision) is a spanning
edge-capacitated tree whose edges are all encoding. In this case, it
is an exercise to prove that any minimum cut can be found as follows.
For $s,t \in V$ we have that
$\lambda(s,t)=\min\{ c'(e): e \in T(st) \}$, where $T(st)$ denotes the
unique path joining $s,t$ in $T$.




In some applications we only specify a subset $Z \subseteq V$ for
which we need cut information. We refer to $Z$ as the {\em terminals}
of the instance. The Gomory-Hu method allows one to store a compressed
version of the GH Tree which only captures cut values $\lambda(s,t)$
for $s,t \in Z$. Namely, a GH $Z$-Tree has $V(T)=Z$.

It is well-known that there may not always exist a GH tree which is a
subgraph of $G$.  For instance, every GH tree for the vertices of
$K_{3,3}$ is a $5$-star (cf. \cite{schrijver2003combinatorial}).  Our
first main result characterizes the graphs which admit GH subtrees.
More precisely, we say that $G$ has the {\em GH Property} if any
subgraph $G'$ of $G$ with any edge-capacity function $c$ has a
Gomory-Hu tree $T$ that is a subgraph of $G'$.

\begin{theorem}\label{thm:1sum}
 $G$ has the GH Property if and only if $G$ is
  the $1$-sum of outerplanar and $K_4$ graphs.
\end{theorem}

We then turn our attention to the generalized version where we are
given a graph-terminal pair $(G,Z)$. Let $G$ be endowed with edge
capacities. A {\em GH $Z$-Tree} is then a capacitated tree
$T=(V(T),E(T))$ (cf. \cite{korte2012combinatorial}). Formally, the
vertices of $T$ form a partition $\{B(v) : v \in Z\}$ of $V(G)$, with
$z \in B(z)$ for all $z \in Z$. Hence Definition~\ref{defn:encoding}
extends as follows.  An edge $B(s)B(t)$ of $T$ is {\em encoding} if
its fundamental cut $(B(S),B(U))$ induces a minimum $st$-cut in
$G$. As before, if all edges are encoding, then $T$ determines the
minimum cuts for all pairs $s,t \in Z$.

We characterize those pairs $(G,Z)$ which admit a GH $Z$-tree as a
minor for any edge capacities on $G$.  We call such a tree a {\em GH
  $Z$-minor} (a formal definition is delayed to
Section~\ref{sec:last}).

Our starting point is the  following elementary observation.
\begin{proposition}\label{prop:notree}
  $K_{2,3}$ has no Gomory-Hu tree that is a subgraph of itself.
\end{proposition}

Even if GH $Z$-minors always existed in a graph $G$, it may still
contain a $K_{2,3}$ minor. The proposition implies, however, that it
should not have a $K_{2,3}$ minor where all nodes in the minor are
terminals. Given a set $Z$ of terminals, we say that $H$ is a {\em
  terminal minor}, or {\em $Z$-minor}, of $G$ if nodes of $V(H)$
correspond to terminals of $G$. In other words, it is a minor such
that each $v \in V(H)$ arises by contracting a connected subgraph
which contains a vertex from $Z$.  Hence a natural necessary condition
for $G$ to always contain GH $Z$-minors is that $G$ must not contain a
terminal-$\ktt$ minor.  We show that this is also sufficient (see
Section~\ref{sec:last} for the formal statement).

\begin{theorem}
\label{thm:minorGH}
Let $Z \subseteq V$.
  $G$ admits a Gomory-Hu tree that is a minor, for
  any capacity function, if and only if $(G,Z)$ is a terminal-$\ktt$
  minor free graph.
\end{theorem}

Establishing the sufficiency requires a better understanding of
terminal minor-free graphs.  We show that the family of pairs $G,Z$
which forbid such terminal-$\ktt$ minors arises precisely as subgraphs
of \emph{$Z$-webs}. $Z$-webs are built from planar graphs with one
outside face which contains all the terminals $Z$ and each inner face
is a triangle to which we may add an arbitrary graph inside connected
to the three vertices; these additional arbitrary graphs are called
\emph{$3$-separated subgraphs}. Subgraphs of $Z$-webs are called {\em
  Extended Okamura-Seymour Instances}.

\begin{theorem}\label{thm:k23char}
  Let $G$ be a $2$-connected terminal-$\ktt$ minor free graph. Then
  either $G$ has at most $4$ terminals or it is an Extended
  Okamura-Seymour Instance.
\end{theorem}

This immediately implies the following.

\begin{corollary}\label{cor:k23char}
  $G$ is terminal-$\ktt$ free if and only if for any $2$-connected
  block $B$, the subgraph obtained by contracting every edge not in $B$
  is terminal-$\ktt$ free.
\end{corollary}

These results also yield the following consequence for multiflow
problems.  Let $G,H$ be graphs such that $V(H) \subseteq V(G)$. Call a
pair $(G,H)$ {\em cut-sufficient} if the cut condition is sufficient
to characterize the existence of a multiflow for any demands on edges
of $H$ and any edge capacities on $G$. If $Z \subseteq V(G)$, we also
call $(G,Z)$ {\em cut-sufficient} if $(G,H)$ is cut-sufficient for any
simple graph on $Z$.

\begin{corollary}\label{cor:flows}
  $(G,Z)$ is cut-sufficient if and only if it is terminal-$\ktt$ free.
\end{corollary}

One can compare this to results of Lomonosov and Seymour
(\cite{lomonosov1985combinatorial,seymour1980four}, cf. Corollary
72.2a \cite{schrijver2003combinatorial}) which characterize the class
of demand graphs $H$ such that every supply graph $G$ ``works'',
i.e. for which $(G,H)$ is cut-sufficient for {\em any} graph $G$ with
$V(H) \subseteq V(G)$.  They prove that any such $H$ is (a subgraph
of) either $K_4, C_5$ or the union of two stars.  A related question
asks for which graphs $G$ is it the case that $(G,H)$ is
cut-sufficient for every $H$ which is a subgraph of $G$; Seymour
\cite{seymour1981matroids} shows that this is precisely the class of
$K_5$ minor-free graphs. We refer the reader to \cite{chekuri2013flow}
for discussion and conjectures related to cut-sufficiency.

The paper is structured as follows. In the next section we prove that
every outerplanar instance has a GH tree which is a subgraph. In
Section~\ref{sec:1sum} we present the proof of
Theorem~\ref{thm:1sum}. In Section~\ref{sec:k23char} we provide the
proofs for Theorem~\ref{thm:k23char} and Corollary~\ref{cor:flows}.
Section~\ref{sec:last} wraps up with a proof of
Theorem~\ref{thm:minorGH}.

\subsection{Some Notation and a Lemma}
\label{sec:notation}

\begin{figure}
  \label{fig:middle}
  \begin{center}
    \begin{tikzpicture}[x=0.35cm,y=0.35cm]
      \fill[color=red,nearly transparent]  plot[smooth cycle, tension=.7]
      coordinates {(2.6,4.6) (0.2,4.8) (-1.8,4.4)
        (-3.4,3.4) (-4.4,1.8) (-3.8,-0.2)
        (-2.7,0.2) (-1.6,2.4) (0.2,3.2)
        (2,3.2) (3.8,3.2) (3.8,4.2)};
      \fill[color=blue, nearly transparent]  plot[smooth cycle, tension=.7]
      coordinates {(-2.6,3.0) (-2,2.6) (-1,2.6)
        (1,3)  (4.3,2.3) (5.8,0.8)
        (7,0.6) (7,2.1) (4.8,4.2)
        (2,5) (-0.6,5) (-2,4.6) (-2.6,3.8)};

      \draw (1.5,1) ellipse (5 and 3);

      \draw [dashed] plot[smooth cycle, tension=.7]
      coordinates {(-0.1,-1.4) (0,-2.2) (-2.2,-1.4)
        (-3.2,-0.6) (-2.6,-0.2) (-1.4,-0.8)};
      \draw [dashed] plot[smooth cycle, tension=.7]
      coordinates {(-3,0.2) (-3.6,0.2) (-4,1.8)
        (-3,2.8) (-2.4,2.4) (-3,1.2)};
      \draw [dashed] plot[smooth cycle, tension=.7]
      coordinates {(1.4,-1.8) (1.4,-2.4) (3.4,-2.4)
        (5,-1.6) (6.6,-0.2) (6,0.4)
        (4.8,-0.6) (3.2,-1.4)};
      \draw [dashed] plot[smooth cycle, tension=.7]
      coordinates {(4.4,3.8) (4.2,3.2) (5.2,2.6)
        (6,1.8) (6,1) (6.8,0.8) (6.8,2)
        (5.6,3.2)};
      \draw [dashed] plot[smooth cycle, tension=.7]
      coordinates {(0.6,4.4) (-0.8,4.2) (-2,3.8)
        (-2.4,3.4) (-1.8,2.8) (-0.4,3.4)
        (1.4,3.6) (3.4,3.4) (3.4,4)
        (3,4.2) (1.6,4.4)};

      \node[bigvertex] (v) at (0.6,-2) {};
      \node[smallvertex] (x2) at (-3.2,2) {};
      \node[smallvertex] (x4) at (5.65,2.65) {};
      \node[smallvertex] (x1) at (-0.3,-1.8) {};
      \node[smallvertex] (x5) at (1.6,-2) {};

      \draw (v) -- (x4);

      \draw  plot[smooth, tension=.7]
      coordinates {(0.6,-2) (-0.8,0.8) (-3.2,2)};

      \path (v) node[anchor=north] {\Large $v$};
      \node at (-2.6,-2) {\large $X_1$};
      \node at (0.6,6) {\large $X_3$};
      \node at (6.8,3.4) {\large $X_4$};
      \node at (-5,2.8) {\large $X_2$};
      \node at (4.6,-2.6) {\large $X_5$};

    \end{tikzpicture}
  \end{center}

\end{figure}

We always work with connected graphs and usually assume (without loss
of generality) that the edge capacities $c(e)$ have been adjusted so
that no two cuts have the same capacity.\footnote{This can be achieved
  in a standard way by adding multiples of $2^{-\delta}$ where
  $\delta = O(|E|)$.} In particular, the minimum $st$-cut is unique
for any vertices $s,t$. Moreover, we may assume any minimum cut
$\delta(X)$ to be {\em central}, a.k.a a {\em bond}. That is,
$G[X],G[V \setminus X]$ are connected.  For any $X \subseteq V(G)$ we
use shorthand $c(X)$ to denote the capacity of the cut $\delta(X)$,
and if $Y \subseteq V(G)$, then $d(X,Y)$ denotes the sum of capacities
for all edges with one endpoint in $X$, and the other in $Y$. We
consistently use $c'(e)$ to denote the computed capacities on edges
$e$ in some Gomory-Hu tree.

As we use the following lemma several times throughout we introduce it
now.

\begin{lemma}
  \label{lem:middle}
  Let $t \in V(G)$ and $X,Y$ be disjoint subsets which induce
  respectively a minimum $xt$-cut and a minimum $yt$-cut where
  $x \in X, y \in Y$.  For any non-empty subset $M$ of $V$ which is
  disjoint from $X \cup Y \cup \{t\}$, we have
  $d(M,V \setminus (X \cup Y \cup M)) > 0$.
\end{lemma}

\begin{proof}
  We have
  \begin{align*}
    &\quad c(M \cup X) + c(M \cup Y) \\
    &= c(X) + c(Y)
      + 2d(M, V \setminus (X \cup Y \cup M))\\
    &< c(M \cup X) + c(M \cup Y)
      + 2d(M, V \setminus (X \cup Y \cup M))
  \end{align*}
  \noindent
  where the second inequality follows from the fact that $\delta(M \cup X)$
  (respectively $\delta(Y \cup M)$) separates $t$ from $X$ (respectively $Y$)
  but  $M \cup X \neq X$ (respectively $M \cup Y \neq Y$).
\end{proof}

\section{Outerplanar graphs have Gomory-Hu Subtrees}\label{sec:gh-tree}
\label{sec:outerplanar}

\begin{theorem}\label{th:gh-subtrees}
  Any $2$-connected outerplanar graph $G$ has a Gomory-Hu tree that is
  a subgraph of $G$.
\end{theorem}

\begin{proof}

  Let $G$ be an outerplanar graph with outer cycle
  $C = v_1, v_2, \ldots, v_n$.  As discussed in
  Section~\ref{sec:notation}, we assume that no two cuts have the same
  capacity, so let $T$ be the unique Gomory-Hu tree of $G$. We want to
  prove that $T$ is a subgraph of $G$.

  Notice that the shore of any min-cut in $G$ must be a subpath
  $v_i,v_{i+1},\ldots,v_{j-1},v_j$ (indices taken modulo $n$) because
  we may assume any min-cut $\delta(S)$ to be {\em central} (a.k.a. a
  {\em bond}), that is, both $S$ and $V-S$ induce connected subgraphs.

  \begin{figure}
    \label{fig:ordering}
    \begin{center}
      \begin{tikzpicture}[x=0.4cm,y=0.4cm]

        \draw (1.5,1) ellipse (5 and 3);

        \node[bigvertex] (v1) at (0.6,-2) {};
        \node[smallvertex] (x2) at (-3.2,2) {};
        \node[smallvertex] (x3) at (0.8,4) {};
        \node[smallvertex] (x4) at (5.65,2.65) {};
        \node[smallvertex] (x1) at (-0.1,-1.9) {};
        \node[smallvertex] (x5) at (1.6,-2) {};

        \draw (v1) -- (x3);
        \draw (v1) -- (x4);
        \draw (v1) .. controls (-0.8,0.5) ..  (x2);

        \draw  plot[smooth cycle, tension=.7]
           coordinates {(0.2,-1.6) (-0.2,-2.3) (-2.2,-1.4)
                        (-3.2,-0.6) (-2.6,-0.2) (-1.4,-0.8)};
        \draw  plot[smooth cycle, tension=.7]
           coordinates {(-3,0.2) (-3.6,0.2) (-4,1.8)
                        (-3,2.8) (-2.4,2.4) (-3,1.2)};
        \draw  plot[smooth cycle, tension=.7]
           coordinates {(1.4,-1.8) (1.4,-2.4) (3.4,-2.4)
                        (5,-1.6) (6.6,-0.2) (6,0.4)
                        (4.8,-0.6) (3.2,-1.4)};
        \draw  plot[smooth cycle, tension=.7]
           coordinates {(4.4,3.8) (4,3.2) (5.2,2.6)
                        (6,1.8) (6,1) (6.8,0.8)
                        (6.8,2) (5.6,3.2)};
        \draw  plot[smooth cycle, tension=.7]
           coordinates {(0.6,4.4) (-0.8,4.2) (-2,3.8)
                        (-2.4,3.4) (-2,2.6) (-0.4,3.4)
                        (1.4,3.6) (3.6,3.4) (3.8,4)
                        (3,4.2) (1.6,4.4)};

        \path (v1) node[anchor=north] {\Large $v$};
        \node at (-2.6,-2) {\large $X_1$};
        \node at (0.6,5.2) {\large $X_3$};
        \node at (6.2,3.4) {\large $X_4$};

        \node at (-4.6,2.4) {\large $X_2$};
        \node at (4.6,-2.6) {\large $X_5$};

      \end{tikzpicture}
    \end{center}

  \end{figure}

  Let $v$ be any vertex and consider the fundamental cuts associated
  with the edges incident to $v$ in the Gomory-Hu tree.  The shores
  (not containing $v$) of these cuts define a partition
  $X_1,X_2,\ldots X_k$ of $V \setminus \{v\}$ where each $X_i$ is a
  subpath of $C$. We may choose the indices such that
  $v, X_1, \ldots, X_k$ appear in clockwise order on $C$.

  \begin{claim}\label{claim:xXiconnected}
    For each $i \in \{1,\ldots,k\}$, there is an edge in $G$ from $v$ to some
    vertex in $X_i$.
  \end{claim}

  \begin{proof}
    By contradiction, assume there is no edge from $v$ to
    $X_i$. Notice $i \notin \{1,k\}$ because of the edges of $C$. Let
    $j \in \{1,\ldots,i-1\}$ maximum with $d(v,X_j) \neq \emptyset$,
    and let $j' \in \{i+1,\ldots,k\}$ minimum with
    $d(v,X_{j'}) \neq \emptyset$, hence $d(v,M) = \emptyset$ where
    $M:= X_{j+1} \cup X_{j+2} \ldots \cup X_{j'-1}$.  By taking
    $X=X_j,Y=X_{j'},t=v$, Lemma~\ref{lem:middle} implies that
    $d(M,V \setminus (X_j \cup X_{j'} \cup M)>0$.  However,
    outerplanarity and the existence of edges from both $X_j$ and
    $X_{j'}$ to $v$, imply that there is an edge between $v$ and $M$,
    cf. Figure~\ref{fig:usingouterplanar}. This contradicts the choice
    of $i$, $j$ or $j'$.
  \end{proof}

  \begin{figure}
    \label{fig:usingouterplanar}
    \begin{center}
      \begin{tikzpicture}[x=0.4cm,y=0.4cm]

        \draw  (1.5,1) ellipse (5 and 3);

        \node[bigvertex] (v1) at (0.6,-2) {};
        \node[smallvertex] (x1) at (-0.6,-1.7) {};
        \node[smallvertex] (x2) at (-3.2,2) {};
        \node[smallvertex] (x4) at (5.7,2.7) {};
        \node[smallvertex] (x5) at (1.8,-2) {};

        \draw (v1) -- (x4);
        \draw (v1) ..controls (-0.8,0.5) .. (x2);

        \draw[dashed] plot[smooth cycle, tension=.7]
          coordinates {(-0.2,-1.4) (-0.2,-2.2) (-2.2,-1.4)
                       (-3.2,-0.6) (-2.6,-0.2) (-1.4,-0.8)};
        \draw[very thick] plot[smooth cycle, tension=.7]
          coordinates {(-3,0.2) (-3.6,0.2) (-4,1.8)
                       (-3,2.8) (-2.4,2.4) (-3,1.2)};
        \draw[dashed] plot[smooth cycle, tension=.7]
          coordinates {(1.4,-1.8) (1.4,-2.4) (3.4,-2.4)
                       (5,-1.6) (6.6,-0.2) (6,0.4)
                       (4.8,-0.6) (3.2,-1.4)};
        \draw[very thick]  plot[smooth cycle, tension=.7,]
           coordinates {(4.4,3.8) (4.2,3.2) (5.2,2.6)
                        (6,1) (6.8,0.8) (6.8,2) (5.6,3.5)};
        \draw[dashed] plot[smooth cycle, tension=.7]
           coordinates {(0.6,4.4) (-0.8,4.2) (-2,3.8)
                        (-2.4,3.4) (-1.8,2.8) (-0.4,3.4)
                        (1.4,3.6) (3.4,3.4) (3.4,4)
                        (3,4.2) (1.6,4.4)};

        \path (v1) node[anchor=north] {\Large $v$};
        \node at (-2.6,-2) {\large $X_1$};
        \node at (0.6,5) {\large $M:=X_3$};
        \node at (6.2,3.4) {\large $X_4$};
        \node at (-5,2.8) {\large $X_2$};
        \node at (4.6,-2.6) {\large $X_5$};
        \node at (0.7,2) {?};

        \draw [color=red] (0.6,3.6) -- (0.6,0.8);
        \draw  [color=red] plot[smooth, tension=.7]
           coordinates {(0,3.6) (-0.4,2.4) (-0.8,1.6)};
        \draw  [color=red] plot[smooth, tension=.7]
            coordinates {(1.8,3.8) (2,2.6) (3.4,1.6)};

      \end{tikzpicture}
    \end{center}

\end{figure}

Let $xy \in E(T)$ be an edge of the Gomory-Hu tree. We must prove that
$xy \in E(G)$. Let $\delta(X)$ be the fundamental cut associated with
$xy$, with $x \in X$, define $Y = V \setminus X$. As in the preceding
arguments we may use the fundamental cuts associated to edges incident
to $x$ and partition $X \setminus \{x\}$ into min-cut shores
$X_1,X_2,\ldots, X_k$; we do this by ignoring the one shore $Y$.
Similarly, we may partition $Y \setminus \{y\}$ into min-cut shores
$Y_1,Y_2,\ldots,Y_l$. We can label these so that
$X_1,X_2,\ldots,X_k,Y_1,\ldots,Y_l$ appear in clockwise order around
$C$ - see Figure~\ref{fig:shores}. There is also some
$i \in \{1,\ldots,k\}$ and $j \in \{1,\ldots,l\}$ such that $x$ is
between $X_i$ and $X_{i+1}$ (or $Y_1$ if $i = k$) and $y$ is between
$Y_j$ and $Y_{j+1}$ (or $X_1$ if $j = l$).

  \begin{figure}[htbp]
    \label{fig:shores}
    \begin{center}
      \begin{tikzpicture}[x=0.6cm,y=0.6cm]

        \draw  (1.6,0.9) ellipse (5 and 3);

        \node[bigvertex] (x) at (0.6,-2.03) {};
        \node[smallvertex] (x4) at (-2.2,-1) {};
        \node[smallvertex] (x5) at (-3.1,2) {};
        \node[smallvertex] (x1) at (4.3,-1.6) {};
        \node[bigvertex] (y) at (3.6,3.7) {};
        \node[smallvertex] (y1) at (2.6,3.85) {};
        \node[smallvertex] (y2) at (4.5,3.4) {};
        \node[smallvertex] (x3) at (-0.1,-1.9) {};
        \node[smallvertex] (x2) at (1.5,-2.1) {};

        \draw (x) .. controls (2.45,-0.8) .. (x1);
        \draw (x) .. controls (-0.8,-0.5) .. (x4);
        \draw (x) .. controls (-1.2,1) .. (x5);

        \draw[thick] plot[smooth cycle, tension=.7]
           coordinates {(-1,-1.2) (-1.4,-1.8) (-2.2,-1.4)
                        (-3.2,-0.6) (-2.6,-0.2) (-1.4,-0.8)};
        \draw [thick] plot[smooth cycle, tension=.7]
           coordinates {(-3,0.2) (-3.6,0.2) (-4,1.8)
                        (-3,2.8) (-2.4,2.4) (-3,1.2)};
        \draw[thick] plot[smooth cycle, tension=.7]
           coordinates {(2.9,-1.9) (3.4,-2.4)
                        (5,-1.8) (6.6,-0.2) (6,0.4)
                        (4.8,-0.6) (3.4,-1.3)};
        \draw[thick] plot[smooth cycle, tension=.7]
           coordinates {(4.4,3.8) (4.2,3.2) (5.5,2.3)
                        (6,1) (6.8,0.8) (6.8,2) (5.6,3.2)};
        \draw[thick] plot[smooth cycle, tension=.7]
           coordinates {(1,4.4) (-1,3.9)
                        (-2.2,3.3) (-1.8,2.8) (-0.4,3.4)
                        (1.4,3.6) (2.5,3.5) (3,4) (2.5,4.5)};
        \draw[thick] plot[smooth cycle, tension=.7]
           coordinates {(-0.6,-1.4) (-1,-2) (0,-2.2) (0,-1.6)};
        \draw[thick] plot[smooth cycle, tension=.7]
           coordinates {(1.2,-1.8) (1.2,-2.2) (2.2,-2.4) (2.2,-1.6)};
        \draw[dashed,color=blue,very thick] plot[smooth, tension=.7]
           coordinates {(-3.7,4) (-1.5,2.3) (2.9,0.6) (7.1,0.7)};

        \path (x) node[anchor=north] {\Large $x$};
        \path (y) node[anchor=south] {\Large $y$};
        \node at (5.4,-2.2) {\large $X_1$};
        \node at (-1,-2.5) {\large $X_3$};
        \node at (-2.8,-1.6) {\large $X_4$};
        \node at (1.8,-3) {\large $X_2$};
        \node at (-4.6,1.6) {\large $X_5$};
        \node at (-1,4.5) {\large $Y_1$};
        \node at (6.5,3.3) {\large $Y_2$};

        \draw (2.6,1.2) -- (2.2,0.1);
        \draw (4.4,1) -- (4.4,0);
        \draw (0.3,2.1) -- (-0.4,1.2);
        \node at (1.8,2) {\LARGE $\delta(X)$};

      \end{tikzpicture}
    \end{center}
    \caption{An arbitrary edge $xy \in T$.}
  \end{figure}
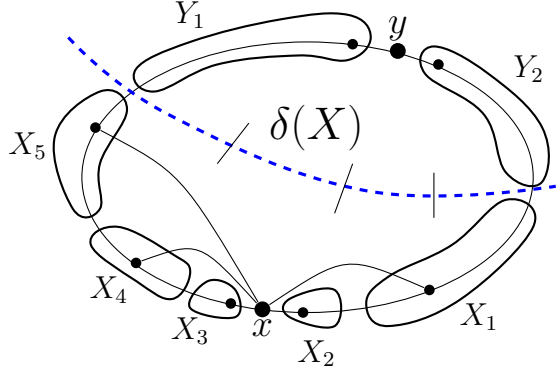

  By contradiction suppose $xy \notin E(G)$. By
  Claim~\ref{claim:xXiconnected}, there is an edge $e$ from $x$ to $Y$,
  let $m \in \{1,\ldots,l\}$ such that $e \in d(x,Y_m)$. If
  $m \notin \{1,l\}$, by outerplanarity either $d(y,Y_1)$ or $d(y,Y_l)$
  is empty; this contradicts Claim~\ref{claim:xXiconnected}. By symmetry
  we may assume $e \in d(x,Y_1)$. By a similar argument there is an edge
  $e' \in d(y,X_1)$. By Claim~\ref{claim:xXiconnected}, there are also
  two edges $e'' \in d(x,X_1)$ and $e''' \in d(y,Y_1)$.

  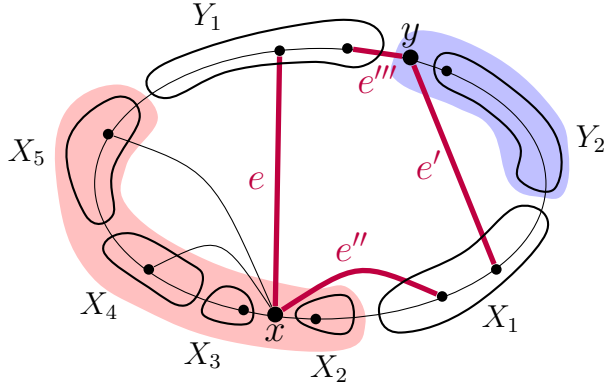
\begin{figure}
    \begin{center}
      \begin{tikzpicture}[x=0.6cm,y=0.6cm]

        \fill[color=red,nearly transparent] plot[smooth cycle, tension=.7]
           coordinates {(-2.9,3.1) (-2,2.4) (-2.7,1.3)
                        (-2.6,0.3) (-0.9,-0.8)
                        (1.2,-1.4) (2.2,-1.4) (2.6,-2.2)
                        (2,-2.8) (-0.4,-2.6)
                        (-2.3,-1.9) (-3.6,-1) (-4.2,0.6)
                        (-4.1,2.4)};
        \fill[color=blue, nearly transparent] plot[smooth cycle, tension=.7]
          coordinates {(3.8,4.3) (3.2,3.6) (3.4,3.1)
                       (4.6,2.6) (5.4,2.1)
                       (5.8,0.9) (6.5,0.6) (7,0.7)
                       (7.1,1.6) (6.8,2.5) (5.5,3.7)};

        \draw  (1.6,0.9) ellipse (5 and 3);
        \node[bigvertex] (x) at (0.6,-2) {};
        \node[bigvertex] (y) at (3.6,3.7) {};
        \node[smallvertex] (x5) at (-3.1,2) {};
        \node[smallvertex] (x4) at (-2.2,-1) {};
        \node[smallvertex] (x3) at (-0.1,-1.9) {};
        \node[smallvertex] (x2) at (1.5,-2.1) {};
        \node[smallvertex] (y1) at (2.2,3.9) {};
        \node[smallvertex] (x1) at (4.3,-1.6) {};
        \node[smallvertex] (y2) at (4.4,3.4) {};
        \node[smallvertex] (y1purple) at (0.7,3.85) {};
        \node[smallvertex] (x1purple) at (5.5,-1) {};

        \draw[line width =2,color=purple]
          (x) .. controls (2.45,-0.8) ..  (x1);
        \node[purple] at (2.3,-0.5) {\Large $e''$};
        \draw[line width =2,color=purple] (x) -- (y1purple)
          node[pos=0.5,left] {\Large $e$};
        \draw[line width =2,color=purple] (y) -- (x1purple)
          node[pos=0.5,left] {\Large $e'$};
        \draw[line width =2,color=purple] (y1) -- (y)
          node[pos=0.5,below] {\Large $e'''$};

        \draw (x) .. controls (-0.6,-0.2) .. (x4);
        \draw (x) .. controls (-0.8,1.2) .. (x5);

        \draw[thick] plot[smooth cycle, tension=.7]
           coordinates {(1,4.4) (-1,3.9)
                        (-2.2,3.3) (-1.8,2.8) (-0.4,3.4)
                        (1.4,3.6) (2.5,3.5) (3,4) (2.5,4.5)};
        \draw[thick] plot[smooth cycle, tension=.7]
           coordinates {(2.9,-1.9) (3.4,-2.4)
                        (5,-1.8) (6.6,-0.2) (6,0.4)
                        (4.8,-0.6) (3.4,-1.3)};
        \draw[thick] plot[smooth cycle, tension=.7]
          coordinates {(-1,-1.2) (-1.4,-1.8) (-2.2,-1.4)
                       (-3.2,-0.6) (-2.6,-0.2) (-1.4,-0.8)};
        \draw[thick] plot[smooth cycle, tension=.7]
          coordinates {(-3,0.2) (-3.6,0.2) (-4,1.8)
                       (-3,2.8) (-2.4,2.4) (-3,1.2)};
        \draw[thick] plot[smooth cycle, tension=.7]
          coordinates {(4.4,3.8) (4.2,3.2) (5.2,2.6)
                       (6,1.8) (6,1) (6.8,0.8)
                       (6.8,2) (5.6,3.2)};
        \draw[thick] plot[smooth cycle, tension=.7]
          coordinates {(-0.6,-1.4) (-1,-2) (0,-2.2) (0,-1.6)};
        \draw[thick] plot[smooth cycle, tension=.7]
          coordinates {(1.2,-1.8) (1.2,-2.2) (2.2,-2.4) (2.2,-1.6)};

        \path (x) node[anchor=north] {\Large $x$};
        \path (y) node[anchor=south] {\Large $y$};
        \node at (5.6,-2.1) {\large $X_1$};
        \node at (-1,-2.9) {\large $X_3$};
        \node at (-3.3,-1.8) {\large $X_4$};
        \node at (1.8,-3.2) {\large $X_2$};
        \node at (-4.9,1.6) {\large $X_5$};
        \node at (-0.9,4.6) {\large $Y_1$};
        \node at (7.6,1.9) {\large $Y_2$};

      \end{tikzpicture}
    \end{center}
    \caption{Showing that $xy \in T$ must be an edge of $G$.}
  \end{figure}

  Let $X' = \{x\} \cup X_2 \cup \ldots \cup X_k$ and
  $Y' = \{y\} \cup Y_2 \cup \ldots \cup Y_l$, $\delta(X')$ is a cut
  separating $x$ from $X_1$ and similarly $\delta(Y')$ separates $y$
  from $Y_1$. As $\delta(X_1)$ is the fundamental cut between $x$ and
  $X_1$, we have that $c(X_1) < c(X')$, and similarly
  $c(Y_1) < c(Y')$. Now, because of the edges $e, e', e'', e'''$, by
  outerplanarity there is no edge between $X'$ and $Y'$, hence
  $$c(X_1) + c(Y_1) = c(X') + c(Y') + 2c(X_1,Y_1) > c(X_1) + c(Y_1) + 2c(X_1, Y_1)$$
  a contradiction.
\end{proof}

\section{Which Instances have Gomory-Hu Subtrees?}
\label{sec:1sum}

The previous result leads to a characterization of graphs with the
{\em GH Property}: that is, graphs whose capacitated subgraphs always
contain a Gomory-Hu Tree as a subtree.  In Section~\ref{sec:last}, we
extend this result to the case where a subset of terminals is
specified.

We start with a simple observation that $\ktt$ does not have a GH
subtree.

\begin{repproposition}{prop:notree}
  $\ktt$, when all edges have capacity $1$, has no Gomory-Hu tree
  that is a subgraph of itself.
\end{repproposition}

\begin{proof}
  Let $\{u_1,u_2\}, \{v_1, v_2,v_3\}$ be the bipartition.  Since the
  minimum $u_1,u_2$ cut is of size $3$, a GH tree should contain a
  $u_1u_2$ path all of whose edges have capacity at least $3$. Suppose
  this path is $u_1v_1u_2$, then the tree's fundamental cut associated
  with $u_1v_1$ should be a minimum $u_1v_1$-cut. But this is impossible
  since $\delta(v_1)$ is a cut of size $2$.
\end{proof}

This leads to the desired characterization.

\begin{reptheorem}{thm:1sum}
$G$ has the GH Property if and only if $G$ is
  the $1$-sum of outerplanar and $K_4$ graphs.
\end{reptheorem}

\begin{proof}
  First suppose that $G$ is such a $1$-sum.  Each outerplanar block in
  this sum has the GH Property by Theorem~\ref{th:gh-subtrees}.  So
  consider a $K_4$ block and a subgraph $G'$ with edge capacities. If
  $G'$ is $K_4$, then clearly any GH tree is a subtree.  Otherwise
  $G'$ is a proper subgraph of $K_4$ and hence is outerplanar.  It
  follows that each block has the GH Property.  It is not hard to see
  that the $1$-sum of Gomory-Hu trees of two graphs is a Gomory-Hu
  tree of the $1$-sum of the graphs. Repeating this argument to the
  blocks we find that $G$ itself satisfies the GH property.

  Suppose now that a $2$-connected graph $G$ has the GH property.  By
  our proposition, $G$ has no $\ktt$ minor.  Outerplanar graphs are
  graphs with forbidden minors $\ktt$ and $K_4$. Hence if $G$ is
  not outerplanar, then it has a $K_4$ minor. Notice that any proper
  subdivision of $K_4$ contains a $\ktt$, as well as any graph
  built from $K_4$ by adding a path between two distinct
  vertices. Hence $G$ must be $K_4$ itself.  The result now follows.
\end{proof}

\section{Characterization of terminal-$\ktt$ free graphs}
\label{sec:k23char}

In this section we prove Theorem~\ref{thm:k23char}. Throughout, we
assume we have an undirected graph $G$ with terminals
$Z \subseteq V(G)$.  We refer to $G$ as being $H$-terminal free (for
some $H$) to mean with respect to this fixed terminal set $Z$.

We first check sufficiency of the condition of
Theorem~\ref{thm:k23char}.  Any graph with at most $4$ terminals is
automatically terminal-$\ktt$ free and one easily checks that any
extended Okamura-Seymour instance cannot contain a terminal-$\ktt$
minor.  Hence we focus on proving the other direction: any
terminal-$\ktt$ minor-free graph $G$ lies in the desired class. To
this end, we assume that $|Z| \geq 5$ and we ultimately derive that
$G$ must be an extended OS instance.

We start by excluding the existence of certain $K_4$ minors.

\begin{proposition}
  If $|Z| \geq 5$ and $G$ has a terminal-$K_4$ minor, then $G$ has a
  terminal-$\ktt$ minor.
\end{proposition}

\begin{proof}
  Let $K_4^+$ be the graph obtained from $K_4$ by subdividing one of
  its edges. By removing the edge opposite to the subdivided edge, we
  see that $K_4^+$ contains $\ktt$. Hence it suffices to prove that
  $G$ contains a terminal-$K_4^+$ minor.

  Consider a terminal-$K_4$ minor on terminals
  $T' = \{s,t,u,v\}$. Thus we have vertex-disjoint trees $T_x$ for
  each terminal $x \in T'$, such that for any $x,y \in T'$, there is
  an edge $e_{xy}$ having one extremity in $T_x$ and one in $T_y$. We
  may assume that $T_x = \bigcup_{y \in T' \setminus \{x\}} P[x,y]$,
  where $P[x,y]$ is a path from $x$ to an end of $e_{xy}$ and not
  containing $e_{xy}$. Denote $U := \bigcup_{x \in T'} V(T_x)$. Assume
  that the minor has been chosen so that $|U|$ is minimized.

  As $|T| \geq 5 > |T'|$, there is some terminal $w \not\in T'$.
  First, suppose that $w$ is contained in one of the subtrees, say
  $T_s$, of the representation of the terminal-$K_4$ minor. Note that
  $w$ could not lie in all three of the paths $P[s,u]$, $P[s,v]$,
  $P[s,t]$ since we could obtain a smaller terminal-$K_4$ minor by
  replacing $s$ by $w$. If $w$ lies in exactly one of these paths, say
  $P[s,u]$, then we obtain a terminal-$K_4^+$ minor where $w$ is the
  terminal which subdivides the minor edge $su$. The last case is
  where $w$ lies in exactly $2$ of the paths, say $P[s,u],P[s,v]$. In
  this case, one may replace $s$ and use $w$ as the degree $3$ vertex
  of the $K_4$ minor and hence $s$ can play the role of the degree $2$
  vertex in a terminal-$K_4^+$ minor, a contradiction.

  Now we assume that $w$ is not contained in any of the subtrees.
  Then by $2$-connectivity, there are two disjoint paths from $w$ to
  two vertices $a$ and $b$ in $U$. If $a$ and $b$ are in different
  subtrees, we easily get a terminal-$K_4^+$ minor, a contradiction.
  Assume $a,b \in T_s$. Now suppose that one of $a$ and $b$ is in
  exactly one of $P[s,u]$, $P[s,v]$, $P[s,t]$, say $P[s,u]$. If both
  $a,b$ have this property, then choose the one which is closer to the
  edge $e_{su}$. Let this be $a$. Then by contracting the $wa$-path we
  get a terminal-$K_4^+$ minor (where $w$ has degree $2$).  Next
  assume that $a$ lies in precisely two of these three paths, say
  $P[s,u],P[s,v]$. Then we also get a terminal-$K_4^+$ where $s$ is now a
  degree $2$ terminal vertex on a subdivided edge between $w$ and $t$.

  In the last case we may assume that $a, b \in R$ where
  $R := P[s,u] \cap P[s,v] \cap P[s,t]$. Let $z$ be the end of $R$
  that is not $s$, and denote $U' := U \setminus (V(R) \setminus z)$.
  Let $Q_1,Q_2$ be openly vertex-disjoint paths from $s$ to
  $U'$. Without loss of generality $Q_1$ contains a vertex $z' \in R$
  which is closest to $z$ amongst all vertices in $Q_1 \cup Q_2$. Hence
  replacing $Q_1$ by the path which follows the subpath of $R$ from
  $z'$ to $z$ also produces a vertex-disjoint pair of paths. Hence we
  assume the endpoints of $Q_1$ are $s$ and $z$. Now consider
  following one of the paths from $w$ until it first hits a vertex of
  $Q_1 \cup Q_2$.  If there is no such vertex, then it hits $R$ and we
  may follow $R$ until it hits $Q_1 \cup Q_2$. In all cases this
  produces a path from $w$ to $U'$ which is disjoint from exactly one
  of $Q_1,Q_2$. Let $P',Q'$ be the resulting vertex-disjoint paths and
  note that one of them terminates at $z$; without loss of generality
  $P'$. Hence its other endpoint can now play the role of $s$ as a
  degree $3$ vertex in a terminal-$K_4$ minor. Therefore, the terminal
  on $Q'$ has a path to $U'$ which is disjoint from $P'$. Thus we are
  back to one of the previous cases.
\end{proof}

Now we have ruled out the existence of terminal-$K_4$ minors, we start
building up minors which can be possible.

\begin{proposition}\label{prop:2conn-terminal-minor}
  Any $2$-connected graph with terminals $Z$, with $|Z| \geq 3$, has a
   $2$-connected minor $H$ with $V(H)=Z$.
\end{proposition}

\begin{proof}
  Let $H$ be a minimal $2$-connected terminal-minor of $G$ containing
  $Z$ and assume there is a non-terminal vertex in $H$. In particular
  we may assume there is an edge $sv$ with $s \in Z$, $v \notin Z$. By
  minimality, contracting $sv$ decreases the connectivity
  to $1$. Hence,  $\{s,v\}$ is a cut separating two
  vertices $t$ and $t'$. Thus, there are two disjoint $tt'$-paths, one
  containing $s$ and the other $v$. That is, there is a circuit $C$
  containing $s,t,v,t'$ in that order.

  By minimality of $H$, we also have that $H-sv$ is not
  $2$-connected. It follows that $H-sv$ contains a cut vertex $\{z\}$
  where $s,v$ lie in distinct components of $H-sv-z$.  This would
  contradict the existence of $C$, and this completes the proof.
\end{proof}

As $|Z| \geq 5$, the previous lemma implies that there is a
terminal-$C_4$ minor. Let $k$ be maximum such that $G$ contains a
terminal-$C_k$ minor.

\begin{proposition}
  $k = |Z|$.
\end{proposition}

\begin{proof}
  By Proposition~\ref{prop:2conn-terminal-minor}, let $H$ be a
  2-connected terminal-minor of $G$ with $V(H)=T$. Consider an
  ear-decomposition of $H$, starting with longest cycle $C_0$ and ears
  $P_1,\ldots, P_k$. Then all ears are single edges (from which the
  proposition follows), otherwise let $P_i$ be an ear that is not a
  single edge, with $i$ minimum. The two ends of $P_i$ are vertices
  $x,y$ of $C_0$. If $x$ and $y$ are consecutive in $C_0$, this
  contradicts the maximality of $C_0$. If they are not consecutive,
  $C_o \cup P_i$ is a subdivision of $\ktt$.

\end{proof}

We let $k=|Z|$ henceforth.  A terminal-$C_k$ minor of $G$ can also be
represented as a collection of $k$ vertex-disjoint subtrees
$T_1, \ldots,T_k$, where each $T_i$ contains exactly one terminal
$t_i$. There also exist edges $e_1,\ldots, e_k$, where $e_i$ has one
extremity $u_i$ in $T_i$ and the other, $v_{i+1}$, in $T_{i+1}$. The
subscript $k+1$ is taken to be $1$; the edges in the subtrees are the
contracted edges and the edges $e_1,\ldots,e_k$ are the undeleted
edges.  We define $s_i$ as the only vertex in
$V(P[t_i,u_i]) \cap V(P[u_i,v_i]) \cap V(P[v_i,t_i])$, where
$V(P[x,y])$ is the vertex set of the path with ends $x$ and $y$ in the
tree $T_i$. Thus, $T_i$ is
$P[s_i,u_i] \cup P[s_i,v_i] \cup P[s_i,t_i]$.


We denote by $S_i$ the path from $t_i$ to $s_i$ in $T_i$ and we take
our representation so that $\sum_{i=1}^k |S_i|$ minimized. We denote
by $P_i$ the path from $s_i$ to $s_{i+1}$.

\begin{proposition}
$\sum_{i=1}^k |S_i| = 0$.
\end{proposition}

\begin{proof}
  By contradiction, suppose $|S_1| > 0$ and so $t_1$ does not lie in the
  graph induced by
  $D=P_1 \cup \ldots \cup P_k \cup S_2 \cup \ldots \cup S_k$.  By
  $2$-connectivity, there are two disjoint minimal paths from $t_1$
  to distinct vertices $x$ and $y$ in $D$.  Moreover we can assume that $x=s_1$
  lies on $P_k \cup P_1$. To see this, suppose that $z \in S_1$ is the closest
  vertex to $s_1$ which is used by the one of the paths (possibly $z=t_1$).
  We may then re-route one of the paths to use the subpath of $S_1$ from $z$ to $s_1$.

  If $y$ is contained in one of $P_k, P_1$, it is routine to
  get another representation of the minor where all the $S_i$ are at
  least as short, and $S_1$ is empty, contradicting the minimality of
  our choice of representation. A similar argument holds if $y \in S_k \cup S_2$.

  So we assume $y \in D \setminus (P_k \cup P_1 \cup S_k \cup S_2)$.
  We now find a terminal-$\ktt$ minor, and that is again a
  contradiction. To see this, let $T_i$ be a tree which contains the
  second vertex $y$. As $k \geq 5$, we may assume either
  $i \in [4,k-1]$, or $i \in [3,k-2]$. Suppose the latter as the two
  cases are similar. We obtain a terminal-$\ktt$ where the two
  degree-3 vertices correspond to the terminals in $T_i$ and
  $T_k$. The degree-2 vertices will correspond to $t_1,t_2$ and
  $t_{k-1}$ --- see Figure~\ref{fig:k23-minor-in-sunny-graph}.
\end{proof}

Hence there is a circuit $C$ containing every terminal, in cyclic
order $t_1$, $t_2$, \ldots $t_k$.

\begin{figure}
  \begin{center}
    \begin{tabular}{cp{1cm}c}
      \begin{tikzpicture}[x=0.7cm,y=0.7cm]
        \draw[line cap=round,line width=0.2cm,color=LightPink]
          (18:2) arc[start angle=18,delta angle=36,radius=1.4cm] (54:2)
          (90:2) arc[start angle=90,delta angle=72,radius=1.4cm] (162:2)
          (162:2) arc[start angle=162,delta angle=72,radius=1.4cm] (234:2)
          (234:2) arc[start angle=234,delta angle=72,radius=1.4cm] (306:2)
          (306:2) arc[start angle=306,delta angle=72,radius=1.4cm] (18:2)
          (18:2) -- (18:3)
          (90:2) -- (0,0)
          (0,0) -- (54:2)
          (162:2) -- (162:3)
          (234:2) -- (234:3)
          (306:2) -- (306:3);
        \draw (0,0) circle[radius=1.4cm];
        \node[smallvertex] (x) at (54:2) {};
        \foreach \theta in {18,90,...,306} {
          \node[smallvertex] (s\theta) at (\theta:2) {};
        }
        \node[terminal] (t90) at (0,0) {};
        \foreach \theta/\i in {18/2,162/k,234/{k-1},306/3} {
          \node[terminal] (t\theta) at (\theta:3) {};
          \path (t\theta) node[anchor=north] {$t_{\i}$};
          \draw (t\theta) -- (s\theta);
        }
        \draw (s90) -- (t90) (t90) -- (x);
        \draw (s90) node[anchor=south] {$s_1$};
        \draw (t90) node[anchor=east] {$t_1$};
        \draw (x) node[anchor=south west] {$y$};
      \end{tikzpicture}
      & &
      \begin{tikzpicture}[x=0.7cm,y=0.7cm]
        \draw[line cap=round,line width=0.2cm,color=LightPink]
          (90:2) arc[start angle=90,delta angle=72,radius=1.4cm] (162:2)
          (162:2) arc[start angle=162,delta angle=72,radius=1.4cm] (234:2)
          (234:2) arc[start angle=234,delta angle=72,radius=1.4cm] (306:2)
          (306:2) arc[start angle=306,delta angle=72,radius=1.4cm] (18:2)
          (18:2) -- (18:3)
          (90:2) -- (90:3)
          (162:2) -- (162:3)
          (234:2) -- (234:3)
          (306:2) -- (306:3)
          (90:3) to[out=330,in=120] (18:2.5);
        \draw (0,0) circle[radius=1.4cm];
        \node[smallvertex] (x) at (18:2.5) {};
        \foreach \theta in {18,90,...,306} {
          \node[smallvertex] (s\theta) at (\theta:2) {};
        }
        \node[terminal] (t90) at (90:3) {};
        \foreach \theta/\i in {18/2,162/k,234/{k-1},306/3} {
          \node[terminal] (t\theta) at (\theta:3) {};
          \path (t\theta) node[anchor=north] {$t_{\i}$};
          \draw (t\theta) -- (s\theta);
        }
        \draw (s90) -- (t90);
        \draw (90:3) to[out=330,in=120] (18:2.5);
        \draw (s90) node[anchor=north] {$s_1$};
        \draw (t90) node[anchor=east] {$t_1$};
        \draw (x) node[anchor=north] {$y$};
      \end{tikzpicture}
      \\
      \begin{tikzpicture}[x=0.7cm,y=0.7cm]
        \draw[line cap=round,line width=0.2cm,color=LightPink]
          (90:2) arc[start angle=90,delta angle=72,radius=1.4cm] (162:2)
          (306:2) arc[start angle=306,delta angle=36,radius=1.4cm] (342:2)
          (18:2) -- (18:3)
          (162:2) -- (162:3)
          (234:2) -- (234:3)
          (306:2) -- (306:3);
        \fill[LightPink] (0,0) circle[radius=6pt];
        \draw (0,0) circle[radius=1.4cm];
        \foreach \theta in {18,90,...,306} {
          \node[smallvertex] (s\theta) at (\theta:2) {};
        }
        \node[terminal] (t90) at (0,0) {};
        \node[smallvertex] (x) at (342:2) {};
        \foreach \theta/\i in {18/2,162/k,234/{k-1},306/3} {
          \node[terminal] (t\theta) at (\theta:3) {};
          \path (t\theta) node[anchor=north] {$t_{\i}$};
          \draw (t\theta) -- (s\theta);
        }
        \draw (s90) -- (t90);
        \draw (t90) -- (x);
        \draw (s90) node[anchor=south] {$s_1$};
        \draw (t90) node[anchor=north east] {$t_1$};
        \draw (x) node[anchor=north west] {$y$};
      \end{tikzpicture}
      & &
      \begin{tikzpicture}[x=0.7cm,y=0.7cm]
        \draw[line cap=round,line width=0.2cm,color=LightPink]
          (90:2) arc[start angle=90,delta angle=72,radius=1.4cm] (162:2)
          (270:2) arc[start angle=270,delta angle=36,radius=1.4cm] (306:2)
          (18:2) -- (18:3)
          (162:2) -- (162:3)
          (234:2) -- (234:3)
          (306:2) -- (306:3);
        \fill[LightPink] (0,0) circle[radius=6pt];
        \draw (0,0) circle[radius=1.4cm];
        \foreach \theta in {18,90,...,306} {
          \node[smallvertex] (s\theta) at (\theta:2) {};
        }
        \node[terminal] (t90) at (0,0) {};
        \node[smallvertex] (x) at (270:2) {};
        \foreach \theta/\i in {18/2,162/k,234/{k-1},306/3} {
          \node[terminal] (t\theta) at (\theta:3) {};
          \path (t\theta) node[anchor=north] {$t_{\i}$};
          \draw (t\theta) -- (s\theta);
        }
        \draw (s90) -- (t90);
        \draw (t90) -- (x);
        \draw (s90) node[anchor=north] {$s_1$};
        \draw (t90) node[anchor=east] {$t_1$};
        \draw (x) node[anchor=north] {$y$};
      \end{tikzpicture}
    \end{tabular}
  \end{center}
  \caption{Reducing $|S_1|$ of finding terminal-$\ktt$ minors
    depending on the position of $y$.}
  \label{fig:k23-minor-in-sunny-graph}
\end{figure}
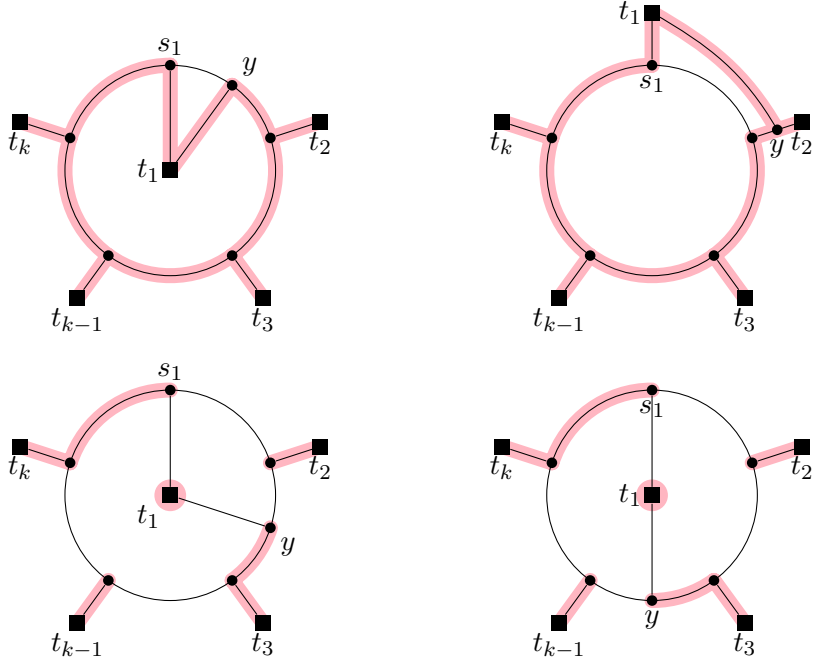

\begin{proposition}\label{prop:no-2-linkage}
  There are no two vertex-disjoint paths, one from $t_i$ to $t_{i'}$, the other
  from $t_{j}$ to $t_{j'}$, with $i < j < i' < j'$.
\end{proposition}

\begin{proof}
  By contradiction. For convenience, let's denote $s = t_i$,
  $t = t_{i'}$, $s' = t_{j}$ and $t' = t_{j'}$. Let $P$ be the
  $st$-path and $Q$ the $s't'$-path. We may assume that we choose $P$
  and $Q$ to minimize their total number of maximal subpaths disjoint
  from $C$.

  We consider the set (not multi-set) of edges
  $E(C) \cup E(P) \cup E(Q)$, and only keep $s, s', t, t'$ as
  terminals.  This defines a subgraph $G'$ of $G$ of maximum degree
  $4$ by construction. Contract edges in $E(C) \cap (E(P) \cup E(Q))$,
  and then contract edges so that vertices of degree 2 are
  eliminated.
  This gives a minor $H$ where the only vertices not of
  degree 4 are $s,t,s',t'$, which have degree 3. $E(H) \cap E(P)$
  induces an $st$-path $P'$ in $H$, $E(H) \cap E(Q)$ induces an
  $s't'$-path $Q'$ in $H$. $P'$ and $Q'$ are again vertex-disjoint. We
  call the remaining edges of $E(C)$ in $H$ {\em $C$-edges}. They
  induce a cycle which alternates between vertices of $P'$ and $Q'$.
  To see this, suppose that $e$ is such an edge joining
  $x,y \in V(P')$ (the case for $Q'$ is the same).  We could then
  replace the subpath of $P$ between $x,y$ by the subpath of $C$ which
  was contracted to form $e$. This would reduce, by at least $1$, the
  number of maximal subpaths of $P$ disjoint from $C$, a
  contradiction.

  Consider the two vertices $u'$ and $v'$ of $Q'$ adjacent to $s$,
  such that $s',u',v',t'$ appear in that order on $Q'$. $u'$ and $v'$
  each has one more incident $C$-edge, whose extremities
  (respectively) are $u$, $v$ and must then be on
  $V(P') \setminus \{s\}$. We create a terminal-$K_4$ minor on
  $s,s',t,t'$ as follows --- see Figure~\ref{fig:k4}, where $u,v$ may
  be in either order on $P'$.
  We contract all the edges of $P'$ except the one $e_s$ incident to
  $s$, and all the edges of $Q'$ except the one $e_{u'}$ incident to
  $u'$ in the direction of $t'$, we get a terminal-$K_4$ minor with
  the edges $su$, $sv$, $uu'$, $vv'$, $e_s$ and $e_u$.  One easily
  checks that this leads to the desired terminal-$K_4$ minor This
  contradiction completes the proof.
\end{proof}

\begin{figure}[htbp]
  \begin{center}
    \begin{tikzpicture}[x=1cm,y=1cm]
      \draw[line cap=round, line width = 0.8cm, color=LightPink]
        (-0.05,2) -- (0.05,2)
        (1,2) -- (11,2)
        (0,0) -- (3,0)
        (4,0) -- (11,0);
      \foreach \i in {0,1,6,8,11} {
        \node[smallvertex] (p\i) at (\i,2) {};
      }
      \foreach \i in {0,3,4,6,11} {
        \node[smallvertex] (q\i) at (\i,0) {};
      }
      \draw[blue,thick]
        (p0) -- (p1) (q3) -- (q4)
        (p0) -- (q3) -- (p6)
        (p0) -- (q6) -- (p8);
      \draw[black]
         (p1) -- (p11)
         (q0) -- (q3)
         (q4) -- (q11);
       \draw (p0) node[anchor=south] {$s$};
       \draw (p6) node[anchor=south] {$u$};
       \draw (p8) node[anchor=south] {$v$};
       \draw (p11) node[anchor=south] {$t$};
       \draw (q0) node[anchor=north] {$s'$};
       \draw (q3) node[anchor=north] {$u'$};
       \draw (q6) node[anchor=north] {$v'$};
       \draw (q11) node[anchor=north] {$t'$};
       \draw (5.5,3) node {$P'$};
       \draw (5.5,-1) node {$Q'$};
    \end{tikzpicture}
  \end{center}
  \caption{How to get a terminal-$K_4$ minor: red parts are contracted into
    single nodes, the blue edges will then form a $K_4$.}
  \label{fig:k4}
\end{figure}
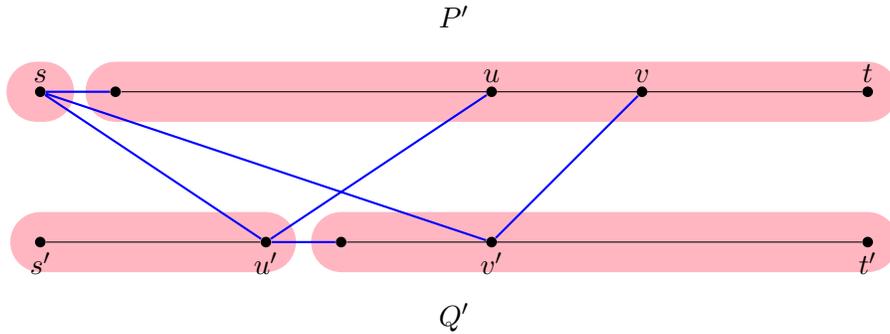









To conclude the characterization of terminal-$\ktt$ minor free
graphs, we use (a generalization of) the celebrated 2-linkage theorem.
Take a planar graph $H$, whose outer face boundary is the cycle
$t_1,t_2,\ldots,t_k$, and whose inner faces are triangles. For each
inner triangle, add a new clique of arbitrary size, and connect each
vertex of the clique to the vertices of the triangle. Any graph built
this way is called a \emph{$(t_1,\ldots,t_k)$-web}, or a
$\{t_1, \ldots ,t_k\}$-web if we do not specify the ordering.

Note that a $Z$-web, for some set $Z$, can be described via
\emph{Okamura-Seymour instances} (OS-instance).  An OS-instance is a
planar graph where all terminals appear on the boundary of the outer
face.  An {\em Extended OS Instance} is obtained from an OS-instance
by adding arbitrary graphs, called \emph{$3$-separated sets}, each
connected to up to three vertices of some inner face of the
Okamura-Seymour instance. We also require that any two $3$-separated
sets in a common face cannot be crossing each other in that
face. Extended OS instances are precisely the $Z$-webs.

\begin{theorem}[Seymour~\cite{seymour1980disjoint},
  Shiloach~\cite{shiloach1980polynomial},
  Thomassen~\cite{thomassen19802}
  ]\label{th:linkage}
  Let $G$ be a graph, and $s_1,\ldots,s_k \in V(G)$. Suppose there are
  no two disjoint paths, one with extremity $s_i$ and $s_{i'}$, and
  one with extremity $s_j$ and $s_{j'}$, with $i < j < i' < j'$.

  Then $G$ is the subgraph of an $(s_1,s_2,\ldots,s_k)$-web.
\end{theorem}

The linkage theorem is usually stated in the special case when
$k = 4$, but the extension presented here is folklore. One can reduce
the general case to the case $k=4$ by identifying the vertices
$s_1,\ldots s_k$ with every other inner vertex of a ring grid with $7$
circular layers and $2k$ rays, and choosing 4 vertices of the outer
layer, labelling them $s,t,s',t'$ in this order, and connecting them
in a square --- see Figure~\ref{fig:linkage-reduction}. Is is easy to
prove that there are two vertex-disjoint paths, one with extremity $s$
and $s'$, the other with extremities $t$ and $t'$ in the graph built
this way if and only there are two disjoint paths as in the theorem in
the original graph (for instance, use the middle layer to route the
path from $s$ to $s_i$ with only 2 bends, then the remaining graph is
a sufficiently large subgrid to route the three other paths). Because
the grid is 3-connected, its embedding is unique and we get that $G$
is embedded inside the inner layer of the ring, from which the general
version of the theorem is deduced.

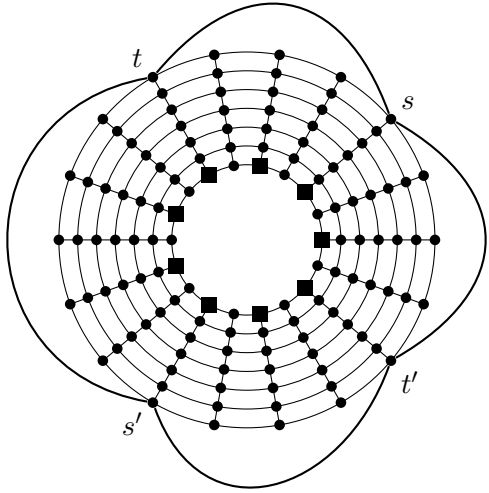
\begin{figure}[htbp]
  \begin{center}
    \begin{tikzpicture}[x=0.25cm,y=0.25cm]
      \foreach \theta in {0,20,...,340} {
        \foreach \distance in {4,5,...,10} {
          \node[smallvertex] (n\theta\distance) at (\theta:\distance) {};
        }
        \foreach \distance in {4,5,...,9} {
          \draw (\theta:\distance) -- +(\theta:1);
        }
      }
      \foreach \theta in {0,40,...,320} {
        \node[terminal] (t\theta) at (\theta:4) {};
      }
      \foreach \distance in {4,5,...,10} {
        \draw (0,0) circle[radius=\distance];
      }
      \draw[thick,looseness=1.5] (n4010) to[out=110,in=50] (n12010);
      \draw[thick,looseness=1.5] (n12010) to[out=190,in=170] (n24010);
      \draw[thick,looseness=1.5] (n24010) to[out=290,in=250] (n32010);
      \draw[thick,looseness=1.5] (n32010) to[out=30,in=330] (n4010);
      \draw (n4010) node[anchor = south west] {$s$};
      \draw (n12010) node[anchor = south east] {$t$};
      \draw (n24010) node[anchor = north east] {$s'$};
      \draw (n32010) node[anchor = north west] {$t'$};
    \end{tikzpicture}
  \end{center}
  \caption{Gadget for the proof of the general linkage theorem.}
  \label{fig:linkage-reduction}
\end{figure}

By using Theorem~\ref{th:linkage} with
Proposition~\ref{prop:no-2-linkage}, we get that any $2$-connected
terminal-$\ktt$ free graph is a subgraph of a $Z$-web where $Z$ is the
set of terminals.

This now completes the proof of
Theorem~\ref{thm:k23char}. \hspace{\stretch{1}} $\Box$

We now establish Corollary~\ref{cor:k23char}.

\begin{proof}
  If $G$ is terminal-$\ktt$ minor-free, then clearly contacting all
  blocks but one must create a terminal-$\ktt$ free instance.

  Conversely, suppose that $G$ has a terminal-$\ktt$ minor. Since
  this minor is $2$-connected, it must be a minor of a graph obtained
  by contracting or deleting all the edges of every $2$-connected
  component except one. Let's call that last block $B$. Hence the
  terminal-$\ktt$ minor is a minor of the graph obtained by
  contracting all the edges not in $B$.
\end{proof}

\subsection{A Consequence for Multiflows}
\medskip

Recall from the introduction that for a graph $G$ and
$Z \subseteq V(G)$, we call $(G,Z)$ cut-sufficient if for any
multi-flow instance (capacities on $G$, demands between terminals in
$Z$), we have feasibility if and only if the cut condition holds.

\noindent
\begin{repcorollary}{cor:flows}
  $(G,Z)$ is cut-sufficient if and only if it is terminal-$\ktt$ free.
\end{repcorollary}

\begin{proof}

  We first establish a lemma which we use again in the next section.

  \begin{lemma}
    Let $G$ be an extended OS instance and $F$ be a $3$-separated
    graph whose attachment vertices to the planar part are
    $\{x,y,z\}$. We may define a new graph $G'$ from $G$ by removing
    $V(F) \setminus \{x,y,z\}$ and add a new vertex $s$ with edges
    $sx,sy,sz$ with capacities $c_x,c_y,c_z$ so that minimum cuts
    separating disjoint sets of terminals in $Z$ have the same
    capacities in $G'$ and in $G$.
  \end{lemma}

  \begin{proof}
    For each $\alpha \in \{x,y,z\}$, let $c_\alpha$ be the value of a
    minimum cut in $F$ separating $\alpha$ from
    $\{x,y,z\} \setminus \{\alpha\}$, for $\alpha \in \{x,y,z\}$. We
    use $S_\alpha$ to denote the shore of such a cut in $F$, where
    $\alpha \in S_{\alpha}$. We replace $H$ in $G$ by a claw where the
    central vertex is a new vertex $u_H$, and leaves are $x$, $y$ and
    $z$, and the capacity of $u_H\alpha$ is $c_\alpha$ for any
    $\alpha \in \{x,y,z\}$. We claim that this transformation
    preserves the values of minimum cuts between sets of terminals.

    Notice that
    $c_\alpha \leq \sum_{\beta \in \{x,y,z\} \setminus \{\alpha\}} c_\alpha$,
    hence a minimum cut $S$ of $G'$ containing $x$ but none of $y$,
    $z$, does not contain $s$. For a cut $S'$ in $G'$ with
    $x \in S, s, y,z \notin S$, we may then associate a cut $S$ in $G$
    with same capacity, by taking $S := S' \cup S_x$. Reciprocally,
    given a cut $S$ of $G$ with $x \in S$, $y, z \notin S$, the cut
    $S' := S \setminus V(F) \setminus \{x,y,z\}$ has capacity at most
    the capacity of $S$. Thus the values of minimum terminal cuts are
    preserved.


  \end{proof}

  Since the $3$-separated graphs are non-crossing, we may iterate the
  process to obtain the following.

  \begin{lemma}\label{lem:upshot}
    For any extended OS instance $G$ we may replace each $3$-separated
    graph by a degree-3 vertex to obtain an equivalent (planar) OS
    instance $G'$. It is equivalent in that for any partition
    $Z_1 \cup Z_2 = Z$, the value of a minimum cut separating $Z_1,Z_2$ in
    $G$ is the same as it is in $G'$.
  \end{lemma}

  We now return to the proof of the corollary.  First, if there is a
  terminal-$\ktt$ minor then we obtain a ``bad'' multiflow instance
  as follows. For each deleted edge we assign it a capacity of
  $0$. For each contracted edge we assign it a capacity of
  $\infty$. The remaining 6 edges have unit capacity. We now define
  four unit demands.  One between the two degree-3 nodes of the
  terminal minor and a triangle on the remaining three nodes.  It is
  well-known that this instance has a flow-cut gap of $\frac{4}{3}$
  cf.~\cite{chekuri2010flow,chakrabarti2012cut}.

  Now suppose that $G$ is terminal-$\ktt$ free and consider a
  multiflow instance with demands on $Z$. By the preceding corollary,
  we may replace each 3-separated graph by a degree-3 vertex and this
  new OS instance will satisfy the cut condition if the old one did.
  Hence the Okamura-Seymour Theorem~\cite{Okamura81} yields a
  half-integral multiflow in the new instance.

  We now show that the flow in the modified instance can be mapped
  back to the original extended OS instance. We do this one
  $3$-separated graph at a time.  Consider the total flow on paths
  that use the new edges through $s$ obtained via the reduction. Let
  $d(xy),d(yz),d(zx)$ be these values. We claim that these can be
  routed in the original $F$. First, it is easy to see that this
  instance on $F$ satisfies the cut condition.  Any violated cut
  $\delta_F(S)$ would contain exactly one of $x,y,z$, say $x$. Hence
  this cut would have capacity less than $d(xy) + d(xz)$ but since
  this flow routed through $s$, this value must be at most $c_x$
  which is a contradiction. Finally, the cut condition is sufficient
  to guarantee a multiflow in any graph if demands only arise on the
  edges of $K_4$, cf. Corollary
  72.2a~\cite{schrijver2003combinatorial}. Hence we can produce the
  desired flow paths in $F$.
\end{proof}

\section{General Case: Gomory-Hu Terminal Trees
            in terminal-$\ktt$ minor free graphs.}
\label{sec:last}

In this section we prove Theorem~\ref{thm:minorGH} using the
characterization of terminal-$\ktt$ minor free graphs. The high
level idea is a reduction to Theorem~\ref{th:gh-subtrees} by
contracting away the non-terminal nodes in the graph.

In the following we let $(G,Z)$ denote a connected graph $G$ and
terminals $Z \subseteq V(G)$.  Recall that the classical Gomory-Hu
Algorithm produces a $GH$ $Z$-Tree $T=(V(T),E(T))$ where formally
$V(T)$ is a partition $\mathcal{P}=\{B(v): v \in Z\}$ of $V(G)$.  We
call $B(v)$ the {\em bag} for terminal $v$ and informally one often
thinks of $V(T)=Z$. In addition each edge $st \in E(T)$ identifies a
minimum $st$-cut in $G$, i.e., it is {\em encoding} as per the
discussion following Definition~\ref{defn:encoding}.

We say that GH $Z$-tree $T$ occurs as a {\em bag minor} in $G$ if (i)
each bag induces a connected graph $G[B(v)]$ and (ii) for each
$st \in E(T)$, there is an edge of $G$ between $B(s)$ and $B(t)$.  We
say that $T$ occurs as a {\em weak bag minor} if it occurs as a bag
minor after deletion of some non-terminal vertices (from its bags and
$G$).

\begin{definition}
  The pair $(G,Z)$ has the {\em GH Minor Property} if for any subgraph
  $G'$ with capacities $c'$, there is a GH Tree which occurs as a bag
  minor in $G'$.  The pair $(G,Z)$ has the {\em weak} GH Minor
  Property if such GH trees occur as a weak bag minor.
\end{definition}


An example where we have the weak but not the (strong) property is for
$\ktt$ where $Z$ consists of the degree $2$ vertices and one of the
degree $3$ vertices, call it $t$. Clearly this is terminal-$\ktt$
minor free since it only has $4$ terminals. The unique GH Tree $T$ is
obtained from $G$ by deleting the non-terminal vertex and assigning
capacity $2$ to all edges in the $3$-star. Hence $T$ is obtained as a
minor (in fact a subgraph) of $G$. However, the bag $B(t)$ consists of
the $2$ degree-$3$ vertices which do not induce a connected subgraph.
Hence $T$ does not occur as a bag minor.  Fortunately, such instances
are isolated and arise primarily due to instances with at most $4$
terminals. We handle these separately.

\begin{proposition}
\label{prop:blech}
Let $G$ be an undirected, connected graph and $Z$ be a subset of at
most $4$ terminals.  Suppose that edge-capacities are given so that no
two central cuts have the same capacity.  Then the unique GH Tree $T$
occurs as a weak bag minor, and if $T$ is a path, then it occurs as a
bag minor.
\end{proposition}

We omit the  to the very end. In the following it is useful to
see how a GH $Z$-Tree bag minor (weak or strong) immediately implies
such a minor for some $Z' \subseteq Z$.

\begin{lemma}
  \label{lem:subset}
  Let $T$ be a GH $Z$-Tree bag minor for some capacitated graph $G$
  and let $v \in Z$.  Let $uv \in T$ be the edge with maximum weight
  $c'(uv)$ in $T$. If we set $B'(u) = B(v) \cup B(u)$ and $B'(x)=B(x)$
  for each $x \in Z \setminus \{u,v\}$, then the resulting partition
  defines a GH $(Z \setminus v)$-Tree $T'$ which is a bag minor.
\end{lemma}

\begin{proof}
  Clearly $T'$ is a bag minor and every fundamental cut of $T$, other
  than $uv$'s, is still a fundamental cut of $T'$. It remains to show
  that for any $a,b \in Z \setminus v$, there is a minimum $ab$-cut
  that does not correspond to the fundamental cut of $uv$.  This is
  immediate if the unique $ab$-path $P$ in $T$ does not contain
  $uv$. If it does contain $uv$, then since $a,b \neq v$, the
  $ab$-path in $T$ contains some edge $vw$.  But since
  $c'(vw) \leq c'(uv)$, the result follows.
\end{proof}

\begin{reptheorem}{thm:minorGH}
  Let $G$ be an undirected graph and $Z \subseteq V$.  $(G,Z)$ has the
  weak GH Minor Property if and only if $(G,Z)$ is a terminal-$\ktt$
  minor free graph. Moreover, if none of $G$'s blocks is a
  $4$-terminal instance, then $(G,Z)$ has the GH Minor Property.
\end{reptheorem}

  \begin{proof}
    If $G$ has a terminal-$\ktt$ minor, then by appropriately
    setting edge capacities to $0,1$ or $\infty$ we find a case where
    $G$ does not have the desired bag minor. So we now assume that $G$
    is a terminal-$\ktt$ minor free graph. Let $G'$ be some subgraph
    of $G$ with edge capacities $c(e)>0$, perturbed so that all
    minimum cuts are unique. We show that the unique GH $Z$-tree
    occurs as a bag minor.

    We deal first with the case where $G'$ has cut vertices. Note that
    one may iteratively remove any leaf blocks which do not contain
    terminals. This operation essentially does not impact the GH
    $Z$-Tree.  Now consider any block $L$. Contracting all other
    blocks into $L$ will put a terminal at each cut vertex in
    $L$. Since such minors must be terminal $\ktt$-free, we may
    actually add all cut vertices to $Z$ in the original graph, and
    the resulting configuration is still $\ktt$-free. Henceforth we
    assume that $Z$ includes these extra vertices and show the desired
    bag minor exists for this terminal set. This is sufficient since
    we can then retrieve the bag minor for the original terminal set
    via Lemma~\ref{lem:subset}.  One checks that a GH $Z$-Tree is
    obtained by gluing together the appropriate GH terminal trees in
    each block. Moreover, since each cut vertex is a terminal, if each
    of these blocks' tree is a bag minor (resp. weak bag minor), then
    the whole tree is a bag minor (resp. weak bag minor).  Therefore
    it is now sufficient to prove the result in the case where $G'$ is
    $2$-connected.

    If $G'$ has at most $4$ terminals, then
    Proposition~\ref{prop:blech} asserts that it has a weak bag minor
    for a GH tree. Moreover, if it has less than $4$ terminals, then
    it's GH Tree is a path and hence occurs as a bag minor.  So we now
    assume that $G'$ contains at least $5$ terminals and hence it is an
    extended OS instance whose outside face is a simple cycle.
    Lemma~\ref{lem:upshot} implies that we may replace each
    $3$-separated set by a degree-$3$ vertex and the resulting graph
    is planar and has the same pairwise connectivities amongst
    vertices in $Z$. It is easy to check that any $Z$-tree bag minor
    in this new graph is also such a minor in the original
    instance. Therefore, it is sufficient to show that any planar OS
    instance with terminals on the outside face has the desired GH
    tree bag minor.

    Denote by $t_1,t_2,\ldots,t_{|T|}$ the terminals in the order in
    which they appear on the boundary of the outer face.
    Let $B(t): t \in Z$ be the bags associated with the (necessarily
    unique) GH $Z$-tree $T$. We show that (i) each $G'[B(t)]$ is
    connected and (ii) for any $st \in T$, there is some edge of $G$
    between $B(s)$ and $B(t)$.

    Consider the fundamental cuts associated with edges incident to
    some terminal $t$.  Let $X_1,X_2, \ldots X_k$ be their shores
    which do not contain $t$.  Since any min-cut is central, each
    $X_i$ intersects the outside face in a subpath of its boundary.
    Hence, similar to Claim~\ref{claim:xXiconnected}
    (cf. Figure~\ref{fig:ordering}), we can order them
    $X_1,\ldots,X_k$ in clockwise order on the boundary with $t$
    between $X_k$ and $X_1$.

    The next two claims complete the proof of the theorem.

    \begin{claim}
      For each terminal $t$, $G'[B(t)]$ is connected.
    \end{claim}

  \begin{proof}
    By contradiction, let $C$ be a component of
    $G' \setminus (X_1 \cup \ldots \cup X_k)$ which does not contain
    $t$.  If $N(C) \subseteq X_i$ for some $i \in \{1,\ldots,k\}$,
    then $\delta(C \cup X_i)$ is a cut separating $t$ from any vertex
    in $X_i$ with capacity smaller than $\delta(X_i)$, contradicting
    the minimality of $X_i$.

    Otherwise choose $j < j'$ be such that
    $N(C) \cap X_j, N(C) \cap X_{j'}$ are non-empty and $j'-j$ is
    maximized.  Call $(j,j')$ the {\em span} of $C$. Without loss of
    generality $C$ has the largest span amongst all components other
    than $B(t)$ (whose span is $(1,k)$ incidentally).  Moreover,
    amongst those (non $B(t)$) components with span $(j,j')$ we may
    assume that $C$ was selected to maximize the graph ``inside'' the
    embedding of $G'[X_j \cup X_{j'} \cup M']$, where
    $M'=C \cup X_{j+1} \cup \ldots X_{j'-1}$.  In particular, any
    component $C'$ with neighbhours in $M$' has
    $N(C') \subseteq M' \cup X_j \cup X_{j'}$.  Let $M$ be the union
    of $M'$ and all such components $C'$. By construction $M$ is
    non-empty and $t \not\in M$, however
    $d(M,V \setminus (X_j \cup X_{j'} \cup M))=0$ which contradicts
    Lemma~\ref{lem:middle} if we take $X=X_j,Y=X_{j'}$.
  \end{proof}

  \begin{claim}\label{claim:shores}
    For each $i \in \{1,\ldots,k\}$, there is an edge from a vertex in
    $B(t)$ to a vertex in $X_i$.
  \end{claim}

  \begin{proof}
    By contradiction, suppose $\delta(B(t),X_i) = \emptyset$, for some
    $i \in \{1,\ldots,k\}$. Let $j$ maximum and $j'$ minimum such that
    $j < i < j'$, $\delta(B(t),X_j) \neq \emptyset$ and
    $\delta(B(t),X_{j'}) \neq \emptyset$. Note that $j$ and $j'$ are
    defined because $X_1$ and $X_k$ are adjacent to $B(t)$ by the
    outer cycle.  If we define $M := X_{j+1} \ldots \cup X_{j'-1}$,
    then $d(M, V \setminus (M \cup X_j \cup X_{j'}))=0$, contradicting
    Lemma~\ref{lem:middle} where we take $X=X_j,Y=X_{j'}$.
  \end{proof}

\end{proof}


Finally, we provide the proof for Proposition~\ref{prop:blech}.

\begin{proof}
  We first consider the case where we have $4$ terminals and let $T$
  be the unique GH tree.  Suppose that $T$ is a star with center
  vertex $1$ and let $B_1,B_2, B_3,B_4$ be the bags.  Since each
  fundamental cut of $T$ is central (in $G$) we have that
  $B_2,B_3,B_4$ each induces a connected subgraph of $G$. Let
  $Y \subseteq B_1$ be those vertices (if any) which do not lie in the
  same component of $G[B_1]$ as $1$.  We may try to produce $T$ as a
  weak bag minor of $G$ by deleting $Y$.  This fails only if for some
  $j \geq 2$, $d(B_1 \setminus Y,B_j)=0$; without loss of generality
  $j=2$.  Let $R=B_2 \cup Y \cup B_3, S=B_2 \cup Y \cup B_4$. It
  follows that $d(R \cap S,V-(R \cup S))=0$ and hence
  $c(R)+c(S)=c(R \setminus S)+c(S \setminus R)=c(B_3)+c(B_4)$.  But
  $\delta(R)$ is a $34$-cut and so $c(R) > c(B_3)$. Similarly,
  $c(S) > c(B_4)$.  But this now contradicts the previously derived
  equality.

  Consider now the case where $T$ is a path, say $1,2,3,4$. Since each
  fundamental cut is central, $G[B_1],G[B_4]$ are connected. Now
  suppose that $G[B_2]$ is not connected.  Let $M$ be the set of
  vertices which do not lie in the same component as $2$.  If we
  define $X=B_1,Y=B_3 \cup B_4$ and $t=2,x=1,y=3$, then
  Lemma~\ref{lem:middle} implies that $d(M,B_2 \setminus M)>0$ a
  contradiction.  It remains to show that $d(B_i,B_{i+1})>0$ for each
  $i=1,2,3$.

  Suppose first that $d(B_1,B_2)=0$. Then
  $c(B_1 \cup B_3 \cup B_4) \leq c(B_3 \cup B_4)$ contradicting the
  fact that $B_3 \cup B_4$ induces the unique minimum $23$ cut.  Hence
  $d(B_1,B_2)>0$ and by symmetry $d(B_3,B_4)>0$. Finally suppose that
  $d(B_2,B_3)=0$.  One then easily checks that
  $c(B_1)+c(B_4) \geq c(B_2)+c(B_3)$.  But then either $B_2$ induces a
  second minimum $12$ cut, or $B_3$ induces another minimum $34$
  cut. In either case, we have a contradiction. The final cases where
  $|Z| \leq 3$ follow easily by the same methods.
\end{proof}

\section{Acknowledgements.}

This paper is dedicated to T.C. Hu for his elegant and fundamental
contributions to combinatorial optimization. Some of this work was
completed during a visit to a thematic semester in combinatorial
optimization at the Hausdorff Research Institute for Mathematics in
Bonn. We are grateful to the institute and the organizers of the
semester.  We thank Bruce Reed who informed us about the reduction
used to derive the general linkage theorem.

\bibliography{conf}{}
\bibliographystyle{plain}

\end{document}